\documentclass[11pt]{amsart}

\usepackage[T1]{fontenc}
\usepackage[utf8]{inputenc}
\usepackage[english]{babel} 
\usepackage{amsmath,amssymb,amsfonts,amsthm}
\usepackage{mathrsfs, esint, dsfont}
\usepackage{moreverb,rotating,graphics,float}
\usepackage{caption, subcaption}
\usepackage[colorlinks, linkcolor={blue}, citecolor={red}]{hyperref}
\usepackage{framed,fancybox}
\usepackage{enumerate}
\usepackage{slashed}
\usepackage{cancel}

\numberwithin{equation}{section} 
\newtheorem{theorem}{Theorem}[section]
\newtheorem{proposition}[theorem]{Proposition}
\newtheorem{remark}[theorem]{Remark}

\newcommand\1{{\mathds 1}}

\newcommand{\eps}{\varepsilon}
\newcommand{\RR}{\mathbb{R}}

\newcommand{\bbI}{\mathbb{I}}

\newcommand{\ri}{\mathrm{i}}
\newcommand{\cS}{\mathscr{S}}
\newcommand{\cD}{\mathcal{D}}
\newcommand{\cE}{\mathcal{E}}

\newcommand{\cC}{\mathcal{C}}

\newcommand{\C}{\mathbb{C}}
\newcommand{\R}{\mathbb{R}}
\newcommand{\N}{\mathbb{N}}
\newcommand{\Z}{\mathbb{Z}}

\newcommand\bA{{\bold A}}

\newcommand\bB{{\bold B}}

\newcommand\be{{\bold e}}

\newcommand\bL{{\bold L}}

\newcommand\bm{{\bold m}}
\newcommand\bp{{\bold p}}
\newcommand{\bP}{\mathbf{P}}

\newcommand\bR{{\bold R}}

\newcommand\bx{{\mathbf{x}}}
\newcommand\by{{\mathbf{y}}}

\def\cB{{\mathcal B}}
\def\cC{{\mathcal C}}
\def\cD{{\mathcal D}}
\def\cE{{\mathcal E}}
\def\cF{{\mathcal F}}
\def\cG{{\mathcal G}}
\def\cH{{\mathcal H}}
\def\cI{{\mathcal I}}

\def\cP{{\mathcal P}}

\def\cS{{\mathcal S}}

\def\cW{{\mathcal W}}

\def\fm{{\mathfrak{m}}}
\def\fh{{\mathfrak{h}}}

\def\ft{{\mathfrak t}}

\def\rd{{\mathrm{d}}}
\def\re{{\mathrm{e}}}
\def\ri{{\mathrm{i}}}

\def\rspan{{\mathrm{span}}}

\newcommand{\rHF}{{\rm rHF}}

\newcommand{\Tr}{{\rm Tr}}

\newcommand{\VTr}{\underline{\rm Tr}}

\newcommand{\curl}{{\bf curl} \,}

\newcommand{\Ran}{{\rm Ran}}
\newcommand{\Ker}{{\rm Ker}}

\newcommand{\norm}[1]{\left\| #1\right\|}
\newcommand{\set}[1]{\left\{ #1\right\}}
\newcommand{\bra}[1]{\left( #1\right)}

\newcommand{\com}[1]{\left[ #1\right]}

\newcommand\ii{{\infty}}

\numberwithin{equation}{section}

\title[DFT for 2-d homogeneous materials with magnetic fields]{Density Functional Theory for two-dimensional homogeneous materials with magnetic fields}

\author{David Gontier}
\address[David Gontier]{CEREMADE, University of Paris-Dauphine, PSL University, 75016 Paris, France \& ENS/PSL University, Département de Mathématiques et Applications, F-75005, Paris, France.}
\email{gontier@ceremade.dauphine.fr}

\author{Salma Lahbabi}
\address[Salma Lahbabi]{EMAMI, LRI, ENSEM, UHIIC, 7 Route d’El Jadida, B.P. 8118 Oasis, Casablanca; 
    MSDA, Mohammed VI Polytechnic University, Lot 660, Hay Moulay Rachid Ben Guerir, 43150,
    Morocco}
\email{s.lahbabi@ensem.ac.ma}

\author{Abdallah Maichine}
\address[Abdallah Maichine]{Department of Mathematics, Faculty of Sciences, Mohammed V University in Rabat, 4 Avenue Ibn Battouta
P.B. 1014 RP, Rabat, Morocco
 \& Ecole Centrale Casablanca, Bouskoura, Ville Verte, P.B. 27182, Morocco
}
\email{abdallahmaichine@gmail.com}

\date{\today}

\begin{document}

\begin{abstract}
    This paper studies DFT models for homogeneous 2D materials in 3D space, under a constant perpendicular magnetic field. We show how to reduce the three--dimensional energy functional to a one--dimensional one, similarly as in our previous work. This is done by minimizing over states invariant under magnetic translations and that commute with the Landau operator. In the reduced model, the Pauli principle no longer appears. It is replaced by a penalization term in the energy. 
\end{abstract}

\maketitle

\tableofcontents


\section{Introduction}
 The analysis of quantum properties of two dimensional materials is an active research area in physics and material science. Some 2D materials such as graphene or phosphorene exhibits many interesting physical properties~\cite{Review-graphene,review-phosphorene, review-MOS2, Kaddar} which has many applications such as High Electron Mobility Transistors \cite{Levetal2018}. Some of these properties are not yet fully understood. This is the main motivation to revisit Density Functional Theory (DFT) when applied to quantum two dimensional systems (see~\cite{BlancLeBris00, GLM21} for previous works). 
 
 As in our previous work~\cite{GLM21}, we study homogeneous two--dimensional slabs, when embedded in three dimensional space, but this time, we include a constant perpendicular magnetic field. We consider a charge distribution $\mu$ which is equidistributed in the  first two dimensions: $\mu(x_1, x_2, x_3) = \mu(x_3)$, and with a constant perpendicular magnetic field $\bB = b \be_3$. 
 
One key result of our previous work was an inequality for the kinetic energy per unit surface for translationally invariant states. Let us quickly summarize the result. Let
\[
    \cP := \left\{ \gamma \in \underline{\mathfrak{S}}^1(L^2(\R^3)) , \quad 0 \le \gamma \le 1, \quad \forall \bR \in \R^2, \quad \tau_\bR \gamma = \gamma \tau_\bR   \right\}
\]
denote the set of one-body density matrices which commute with all $\R^2$ translations. Here, $\underline{ \mathfrak{S}}^1(L^2(\R^3))$ stands for locally trace class self--adjoint operators with finite  {\em trace per unit surface}  $\VTr(\gamma)<\ii$ (see Section~\ref{ssec:representation gamma}). For $\bR \in \R^2 \subset \R^3$, we have denoted by $\tau_\bR f(x_1, x_2, x_3) := f(x_1 - R_1, x_2 - R_2, x_3)$ the usual translation along the first two dimensions. Let us also introduce the set of reduced states 
\[
    \cG := \left\{ G \in \mathfrak{S}^1(L^2(\R)), \quad G \ge 0\right\},
\]
where $\mathfrak{S}^1(L^2(\R))$ refers to the space of trace class self-adjoint operators on $L^2(\R)$. We have proved in~\cite{GLM21} that, for any (representable) density $\rho = \rho(x_3)$ depending only on the third variable, we have
\begin{equation}\label{eq:equat_intro_1}
    \inf_{\gamma \in \cP \atop \rho_\gamma = \rho} \left\{ \frac12 \VTr(- \Delta_3 \gamma) \right\} 
    =  \inf_{G \in \cG \atop \rho_G = \rho} \left\{ \frac12 \Tr( - \Delta_1 G) + \pi \Tr(G^2) \right\} ,
\end{equation}
where $\Delta_d$ denotes the Laplacian operator in $d$--space dimension. The energy appearing in the right hand side leads to one--dimensional reduced models for homogeneous semi-infinite slabs in the context of DFT. One typically obtains a minimization problem of the form
\begin{equation}\label{eq:intro_DFT_noB}
    \inf \left\{ \frac12\Tr(-\Delta_1 G)+\pi\Tr(G^2) + \frac12 \cD_1(\rho_G -\mu) + E^{\rm xc}(\rho_G), \quad G \in \cG  \right\},
\end{equation}
where $\cD_1(\cdot)$ is the one--dimensional Coulomb interaction energy (see~\cite[Section 3.1]{GLM21} for a discussion about this term), and $E^{\rm xc}$ some exchange-correlation energy {\em per unit surface}. Note that there is no Pauli principle for the operator $G$; it has been replaced by the penalization term $\pi \Tr(G^2)$ in the energy, which prevents $G$ from having large eigenvalues. We refer to~\cite{GLM21} for details, where we also studied the reduced Hartree--Fock model, which corresponds to the case $E^{\rm xc} = 0$.

\medskip

The scope of this paper is to apply a similar reduction when taking into  account magnetic effects. Without considering the spin (we refer to Section~\ref{sec:spins} for the case with spin), the kinetic energy {\em per unit surface} is of the form
$$\frac12 \VTr \left(  (- \ri \nabla_3 + \bA)^2 \gamma  \right), $$
where $\bA  = b (0, x_1, 0)$ is a vector potential so that $\curl \bA= \bB = b \be_3$. We chose a gauge which is not symmetric, but which will simplify some computations. The Laplacian operator $-\Delta_3$ has been replaced by the Landau operator 
\begin{equation*} 
   \bL^\bA_3 := (- \ri \nabla_3 + \bA)^2  = \bL^\bA_2 - \partial_{x_3x_3}^2, 
   \quad \text{with} \quad
   \bL^\bA_2 = -\partial_{x_1x_1}^2 + ( - \ri \partial_{x_2} + b x_1)^2.
\end{equation*}
In analogy with \cite{GLM21}, we only consider states commuting with the so--called {\em magnetic translations} $\fm_\bR$ and with the Landau operator. We refer to Section~\ref{sec:reduction_kinetic} for the definition of these operators, and for the justification of this choice. We denote the set of such states by
\begin{equation}\label{eq:def:PA}
    \cP^\bA := \left\{ \gamma \in \underline{\mathfrak{S}}^1(L^2(\R^3)), \ 0 \le \gamma \le 1, \ \forall \bR \in \R^2, \ \fm_\bR \gamma = \gamma \fm_\bR,\; \bL^\bA_2 \gamma=\gamma \bL^\bA_2   \right\}.
\end{equation}
In Theorem~\ref{th:decomposition}, we prove that any $\gamma \in \cP^\bA$ has a simple decomposition, in terms of the different projectors on the Landau levels. Using this decomposition, we prove in Theorem~\ref{th:reduction} that, in the magnetic case, we have an equality similar to~\eqref{eq:equat_intro_1}, which reads
\begin{equation}\label{eq:equat_intro_2}
     \boxed{ \inf_{\gamma \in \cP^\bA \atop \rho_\gamma = \rho} \left\{ \frac12 \VTr( \bL^\bA_3 \gamma) \right\} 
    =  \inf_{G \in \cG \atop \rho_G = \rho} \left\{ \frac12 \Tr( - \Delta_1 G) + \Tr \left( F(b, G) \right) \right\} }
\end{equation}
with the penalization term $F$ defined by 
\[  
    F(b,g) :=  \pi g^2 + \frac{b^2}{4 \pi } \left\{ \frac{2 \pi g}{b}   \right\} \left(1 -  \left\{ \frac{2 \pi g}{b}   \right\}  \right),
\]
where $\{ x \}$ denotes the fractional part of $x$. 
The function $F$ is studied in Proposition~\ref{prop:props-F}. It is a piece-wise linear function, reflecting the contribution of the different Landau levels. The function $F$ is not new, and already appears in the context of the two-dimensional Thomas Fermi (TF) theory under constant magnetic fields (see~\cite{LiebSolYng95} and related references). The (spinless) TF kinetic energy takes the form
\[
    E_{\rm kin}^{\rm TF}(b, \rho) := \int_{\R^2} F(b, \rho(\bx)) \rd \bx.
\]
For $b=0$, we recover the usual two--dimensional TF kinetic energy $\int \pi\rho^2$. It is different from the three--dimensional TF kinetic energy of a gas under a constant magnetic field, which has been derived and studied in~\cite{FuGuPeYn92, Yngvason91}, and is obtained by assuming that the electron density is constant, hence also invariant under the third--direction translation.

\medskip

Equation \eqref{eq:equat_intro_2} allows the reduction of DFT models for two-dimensional homogeneous slabs under constant magnetic field. In fact, one obtains a one-dimensional problem of the form (compare with~\eqref{eq:intro_DFT_noB})
\[
     \inf \left\{ \frac12\Tr(-\Delta_1 G)+\Tr\left( F(b, G) \right) + \frac12 \cD_1(\rho_G -\mu) + E^{\rm xc}(b, \rho_G), \quad G \in \cG  \right\}.
\]
Note that the exchange-correlation function may depend on the external magnetic field $b$ (see \cite[Eqn. (4.1)-(4.2)]{FuGuPeYn92} for the expression of the exchange energy for the Landau gas). In Section~\ref{ssec:rHF}, we study the corresponding reduced Hartree-Fock model, where $E^{\rm xc} = 0$.

\medskip 

This article is structured as follows. In Section \ref{sec:reduction_kinetic}, we start by introducing the Landau operator and studying its spectral decomposition, then we define magnetic translations $\{\fm_\bR\}_{\bR\in\R^2}$ with some of their properties, and we characterize the states in $ \cP^\bA$. Using this characterization, we explain in section~\ref{sec:reduction} how to reduce the kinetic energy, we give some properties of the penalization term $F$ and we study the corresponding reduced Hartree-Fock model. Finally, we show in Section~\ref{sec:spins} how to extend our results to systems with spin.


\subsection*{Acknowledgments}  The research leading to these results has received funding from OCP grant AS70 ``Towards phosphorene based materials and devices''.

\section{States commuting with magnetic translations}
\label{sec:reduction_kinetic}

In this section, we prove that states $\gamma \in \cP^\bA$ have a particular structure.

\subsection{Two dimensional Landau operator}\label{sec:landau}
We start by recalling some classical facts about the (two--dimensional) Landau operator $\bL^\bA_2 = - \partial_{x_1 x_1}^2 + ( - \ri \partial_{x_2} + b x_1)^2$. In what follows, we assume $b \neq 0$. For $n \in \N_0 = \{0, 1, 2, \cdots \}$, we introduce the function $\varphi_n : \R \to \R$ defined by
\begin{equation}\label{def varphin}
    \varphi_n(x) := a_n | b|^{1/4} \cH_n(\sqrt{| b |}x) \re^{-\tfrac12 | b | x^2},
\end{equation}
 where $\cH_n(x)=(-1)^n \re^{x^2}\frac{\rd^n}{\rd x^n}\re^{-x^2}$ refers to the $n$-th Hermite polynomial and $a_n=(2^n n!)^{-1/2}/\pi^{1/4}$ is a normalization constant so that $\| \varphi_n \|_{L^2(\R)} = 1$. 
 
\begin{proposition}\label{prop spectral properties L_A} 
    The operator $\bL^A_2$ has purely discrete spectrum
    \begin{equation}\label{eq:sigmaLA2}
    \sigma(\bL^A_2)=b(2\N_0 + 1).
    \end{equation}
    The eigenvalue $\varepsilon_n := b (2n + 1)$ is of infinite multiplicity, with eigenspace
    \begin{equation}\label{def En (Landau level)}
        E_n := \ker(\bL^\bA_2 - \varepsilon_n )=\{\cW(\varphi_n,g),\; g\in L^2(\RR)\}, 
    \end{equation}
    where $\cW$ is a Wigner type transform defined on $L^2(\RR)\times L^2(\RR)$ by
    \begin{equation}\label{eq:def-W}
        \cW(\varphi,g)(\bx):= \frac{1}{\sqrt{2 \pi}} \int_{\R}\re^{-\ri k x_2} \varphi \bra{x_1-\frac{k}{b}} g(k) \rd k.
    \end{equation}
    In particular, the spectral projector $\bP_n$ onto $\ker(\bL^\bA_2 - \varepsilon_n )$ has kernel
    \begin{equation} \label{form kernel Pn}
        \bP_n(\bx ; \by) = \frac{1}{2 \pi} \int_{\R} \re^{ -\ri k (x_2 - y_2)} \varphi_n \left( x_1 - \frac{k}{b} \right)  \varphi_n \left( y_1 - \frac{k}{b} \right) \rd k.
    \end{equation}
    Its density $\rho_{\bP_n}(\bx) := \bP_n(\bx ; \bx) =  \frac{b}{2 \pi}$ is constant and independent of $n$.
\end{proposition}

\begin{remark}
The definition~\eqref{eq:def-W} is slightly different from the classical Wigner transform (see for example~\cite[Chapter 2]{Wong98}) which is rather adapted to study Landau operator with the gauge $\tilde{\bA}=\frac{b}{2}\begin{pmatrix} -x_2 \\ x_1
\end{pmatrix}$, for $b=1$. A gauge transformation  links the two transforms.

\end{remark}

\begin{proof}
    First, we remark that $\bL^\bA_2$ commutes with all translations in the $x_2$--direction. We introduce the Fourier transform with respect to the $x_2$--variable
    \begin{equation} \label{def Fourier x2}
        \cF[f](x_1, k) := \frac{1}{\sqrt{2 \pi}} \int_{\R}f(x_1,x_2)\re^{\ri k x_2} \rd x_2
    \end{equation}
and its inverse
$$
\cF^{-1}[\phi](x_1, x_2) := \frac{1}{\sqrt{2 \pi}} \int_{\R}\phi(x_1,k)\re^{-\ri k x_2} \rd k. 
$$
 We have 
    \begin{equation} \label{eq reduction of L_A via Fourier}
        \cF \bL^A_2 \cF^{-1}  = \int_{\R}^\oplus \fh_{b,k} \rd k, \quad \text{with} \quad
        \fh_{b,k}:=-\partial_{xx}^2 +(bx-k)^2.
    \end{equation}
    The operator $\fh_{b, k}$ is a translation of the harmonic oscillator $\fh := - \partial_{xx}^2 + b^2 x^2$, whose spectral decomposition is
    \[
    \fh = \sum_{n=0}^\infty \varepsilon_n | \varphi_n \rangle \langle \varphi_n |, 
    \]
    with $\varepsilon_n = b (2n + 1)$ and $\varphi_n$ as defined in~\eqref{def varphin}. We deduce that the spectral decomposition of $\fh_{b, k}$ is
    \[
    \fh_{b,k} := \sum_{n=0}^\infty \varepsilon_n | \varphi_n(\cdot - \tfrac{k}{b}) \rangle \langle \varphi_n( \cdot - \tfrac{k}{b}) |,
    \]
    which proves~\eqref{eq:sigmaLA2}. Using~\eqref{eq reduction of L_A via Fourier}, we see that $W$ is an eigenfunction of $\bL^{\bA}_2$, corresponding to eigenvalue $\varepsilon_n$, if and only if it is of the form 
     $$W=\cF^{-1} \com{ (x_1,k) \mapsto g(k)\varphi_n\bra{x_1-\frac{k}{b}}}=\cW(\varphi_n,g),$$
     where $g\in L^2(\RR)$. Thus 
    \[
        E_n := \ker(\bL^\bA_2 - \varepsilon_n)=  \left\{ \cW(\varphi_n,g),\; g\in L^2(\RR) \right\}.
    \]
  To compute the kernel $\bP_n(\bx, \by)$ of the projector on $E_n$, we use the Moyal identity, that we recall here.
    \begin{proposition}\label{prop Moyal identity}
    	Let $f_1,g_1, f_2, g_2$ in $L^2(\R)$. Then, $\cW(f_1,g_1), \cW(f_2, g_2) \in L^2(\RR^2)$ and {\rm (Moyal identity)}
    	\begin{equation}\label{eq:moyal}
    		\langle \cW(f_1,g_1) , \cW(f_2,g_2) \rangle_{L^2(\R^2)} = \langle f_1 , f_2 \rangle_{L^2(\R)} \, \langle g_1, g_2 \rangle_{L^2(\R)}.
    	\end{equation}
    \end{proposition}
    
    \begin{proof}
    	We first prove the result for $f_1, g_1, f_2, g_2 \in C^\infty_0(\R)$ and conclude by density. Parseval identity gives
    	\begin{align*}
    		\langle \cW\bra{f_1,g_1} , \cW\bra{f_2,g_2} \rangle_{L^2(\R^2)} 
    		&= \iint_{\R^2} (\overline{f_1} f_2) \left( x_1 - \frac{k}{b} \right) (\overline{g_1} g_2) (k) \rd x_1 \rd k \\
    		& = \langle f_1 , f_2 \rangle_{L^2(\R)} \, \langle g_1, g_2 \rangle_{L^2(\R)}.
    	\end{align*}
    \end{proof}
    In particular, if $(\psi_m)_{m \in \N}$ is any basis of $L^2(\R)$, then $\{ \cW(\varphi_n, \psi_m) \}_{m \in \N}$ is a basis of $E_n$.   
     Notice that, for any fixed $\bx \in \R^2$, we have
    \[
    \cW(\varphi, g) (\bx) = \langle \varphi_{\bx}, g \rangle_{L^2(\R)}, \quad \text{with} \quad \varphi_\bx(k) := \frac{1}{\sqrt{2 \pi}}  \re^{ \ri k x_2 } \overline{\varphi} \left( x_1 - \frac{k}{b} \right).
    \]
    Hence
    \begin{align*}
        \bP_n(\bx ; \by) & = \sum_{m \in \N} \cW(\varphi_n, \psi_m)(\bx) \overline{\cW(\varphi_n, \psi_m)}(\by)
        = \sum_{m \in \N} \langle \varphi_{n, \bx}, \psi_m \rangle \langle \psi_m, \varphi_{n,\by} \rangle.
    \end{align*}
    Using that  $\sum_m | \psi_m \rangle \langle \psi_m | = \bbI_{L^2}$, we obtain $ \bP_n(\bx ; \by) = \langle  \varphi_{n, \bx},  \varphi_{n, \by} \rangle$, which, given that $\varphi_n$ is real-valued, gives~\eqref{form kernel Pn}. The density of $\bP_n$ is thus
    \[
    \rho_{\bP_n}(\bx) = \bP_n(\bx ; \bx) =  \frac{1}{2 \pi} \int_{\R} \left| \varphi_n \left( x_1 - \frac{k}{b} \right) \right|^2 \rd k = \frac{b}{2 \pi}.
    \]
\end{proof}

\subsection{Magnetic translations}
The Landau operator does not commute with the usual translations, however it commutes with the magnetic translations, that we define now. We write
\[
    \bL_2^\bA := p_{\bA, 1}^2 + p_{\bA, 2}^2, \quad \text{with} \quad p_{\bA, 1} := - \ri \partial_{x_1}, \quad p_{\bA, 2} := - \ri \partial_{x_2} + b x_1.
\]
The operators $p_{\bA, 1}$ and $p_{\bA, 2}$ do not commute, and do not commute with $L_2^\bA$. Actually, we have
\[
    \left[p_{\bA, 1}, p_{\bA, 2}  \right] = - \ri b, \quad \left[p_{\bA, 1}, L_2^\bA  \right]  = - 2 \ri b p_{\bA, 2}, \quad \left[p_{\bA, 2}, L_2^\bA  \right]  = 2 \ri b p_{\bA, 1}.
\]
However, introducing the dual momentum operators
\[
    \widetilde{p}_{\bA, 1} := - \ri \partial_{x_1} + b x_2 , \quad \widetilde{p}_{\bA, 2} := - \ri \partial_{x_2}, 
\]
we can check that $\left[\widetilde{p}_{\bA, 1}, L_2^\bA  \right] = \left[\widetilde{p}_{\bA, 1}, L_2^\bA  \right] = 0$. The magnetic translation $\fm_\bR$, $\bR \in \R^2$, is the unitary operator 
\begin{align*}
    \fm_\bR & = \exp(-\tfrac{\ri}2 b R_1 R_2) \exp \left( - \ri \widetilde{\bp}_\bA  \cdot \bR \right) \\
    & = \exp(-\tfrac{\ri}2 b R_1 R_2)  \exp \left( - \ri \left(\widetilde{p}_{\bA, 1} R_1 +\widetilde{p}_{\bA, 2} R_2   \right) \right).
\end{align*}
Note that we have added a phase factor in order to match the usual convention. Using the Baker--Campbell--Haussdorf formula and the fact that $\left[ \widetilde{p}_{\bA, 1}, \widetilde{p}_{\bA, 2} \right] = \ri b$ commutes with all operators, we obtain  the explicit expression 
\[
    \fm_\bR =  \exp (- \ri b R_1 x_2)  \tau_\bR, \quad \text{that is} \quad \left( \fm_\bR f \right) (\bx) = \exp (- \ri b R_1 x_2) f(\bx - \bR),
\]
where $\tau_\bR f (\bx) := f(\bx - \bR)$ is the usual translation operator. 
By construction, the magnetic translations commute with $L^\bA_2$ and $\bP_n$, but they do not commute among them. Actually, we have
\[
    \fm_\bR\fm_{\tilde{\bR}}=\re^{\ri b R_2\tilde{R}_1}\fm_{\bR+\tilde{\bR}}
    \quad \text{and} \quad
     \fm_\bR^*=\fm_\bR^{-1}=\re^{\ri b R_1 R_2}\fm_{-\bR}.
\]

An important feature of magnetic translations is that they form an irreducible family on each eigenspace $E_n$, in the sense of~\cite[Definition 2.3.7]{bratteli}. 

\begin{proposition}\label{prop:irreducibility}
	The set of magnetic translation operators $(\fm_\bR)_\bR$ is an irreducible family of operators on each $E_n$, in the sense that
    \begin{equation}\label{eq irreducibility characterization}
    \forall \Psi\in E_n\setminus\set{0}, \quad         E_n=\rspan\{ \fm_\bR\Psi: \bR\in\R^2 \}. 
    \end{equation}
\end{proposition}

\begin{proof}
Assume otherwise, and let $\Psi$ so that 
\[
    \widetilde{E}_n(\Psi) := \rspan\{ \fm_\bR\Psi: \bR\in\R^2 \} \subsetneq E_n.
\]
Then there is $\Phi \in E_n \setminus \{ 0 \}$ so that $\Phi \perp \widetilde{E}_n(\Psi)$.
Let $f, g \in L^2(\R)$ so that $\Psi = \cW(\varphi_n,g)$ and $\Phi = \cW(\varphi_n, f)$. Using the Moyal identity, and the fact that
\[
\fm_\bR \cW(\varphi_n, g) = \cW(\varphi_n, \ft_\bR g), \quad \text{with} \quad 
\ft_\bR g: k\mapsto\re^{-\ri b R_1 R_2} \re^{\ri k R_2} g(k-b R_1),
\]
the condition $\langle \Phi, \bm_\bR \Psi \rangle = 0$ for all $\bR \in \R^2$ reads
\[
\forall \bR \in \R^2, \quad \langle f, \ft_\bR g  \rangle = 0, \quad \text{hence} \quad 
\int_{\R} \overline{f}(k) \re^{ \ri k R_2} g(k - b R_1) \rd k = 0.
\]
Applying the inverse Fourier transform to $k\mapsto \overline{f}(k)g(k-R_1)=0$ shows that $\overline{f}(k) g(k - R_1) = 0$ a.e. for all $R_1 \in \R$. Squaring and integrating in $R_1$ gives $f=0$, a contradiction.
\end{proof}

\subsection{Diagonalisation of states commuting with magnetic translations}
\label{ssec:representation gamma}
In what follows, we are interested in one-body density matrices which commute with all magnetic translations. 
In the case without magnetic field, if a state commutes with all usual translations $\tau_\bR$, then it commutes with the Laplacian operator. In the magnetic case, there are operators which commute with all magnetic translations $(\fm_\bR)_{\bR\in \R^2}$, but which do not commute with the Landau operator (we give an example of such an operator in Remark~\ref{rmk op commuting with m_R but not with L_A} below). So, we rather consider one-body density matrices which commute with $\fm_\bR$, and with the Landau operator. It turns out that such operators have an explicit and simple characterization.

\medskip
 We first enunciate our result in dimension two before turning to the three dimensional case. 

\begin{proposition}\label{prop:decomposition}
    Let $\eta\in \cS(L^2(\R^2))$ be such that $\eta\fm_\bR=\fm_\bR\eta$ for all $\bR\in\R^2$ and $\eta \bL_\bA= \bL_\bA \eta$. Then, there is a family of real numbers $(\lambda_n)_{n \in \N_0}$ so that 
    \begin{equation}\label{eq:decomposition}
        \eta=\sum_{n\in\N_0} \lambda_n\bP_n.	
    \end{equation}
If $\eta$ is a locally trace class operator, then its density is constant 
$$
\rho_\eta= \frac{b}{2\pi}\sum_{n\in\N_0}\lambda_n. 
$$
\end{proposition}

\begin{proof}
    Since $\eta$ commutes with $\bL^\bA_2$, it commutes with any spectral projector $\bP_n$, hence leaves invariant $E_n=\Ran(\bP_n)$, for all $n\in\N$. The operator $\eta_n:=\bP_n\eta\bP_n \in \cS(E_n)$ commutes with all magnetic translations. Since the family $\{ \bm_\bR \}_{\bR}$ is irreducible, it implies that $\eta_n$ is proportional to the identity on $E_n$, hence is of the form $\eta_n = \lambda_n \bbI_{E_n}$. This is a kind of Schur's Lemma, see~\cite[Proposition 2.3.8]{bratteli}.
\end{proof}

\begin{remark}\label{rmk op commuting with m_R but not with L_A}
    The hypothesis $\eta \bL^\bA_2= \bL^\bA_2 \eta$ is not a consequence of the commutation with the operators $(\fm_\bR)_{\bR}$. Indeed, consider for a normalized $\zeta \in L^2(\R)$, the projector $P_\zeta$ onto the vectorial space
    \[
    E_\zeta := \left\{ \cW(\zeta, f), \ f \in L^2(\R) \right\}.
    \]
    Since $\fm_\bR \cW(\zeta, f) = \cW(\zeta, \ft_\bR f) \in E_\zeta$, the set $E_\zeta$ is invariant by $\fm_\bR$, hence $P_\zeta$ commutes with all magnetic translations. However, we have, using the decomposition~\eqref{eq reduction of L_A via Fourier} that
    \begin{align*}
        \bL^\bA_2 \left[ \cW(\zeta, f) \right] & = \frac{1}{\sqrt{2 \pi}} \int_\R \re^{ - \ri k x_2} \left[  ( - \partial_{x_1x_1}^2 +\left( bx_1 - k\right)^2) \zeta\left(x_1 - \frac{k}{b} \right) \right] f(k) \rd k \\
        & = \cW (\widetilde{\zeta}, f)
    \end{align*}
    with $\widetilde{\zeta} := \left( -\partial_{xx}^2 + b^2 x^2  \right) \zeta$. So $\bL^{\bA}_2 E_\zeta = E_{\widetilde{\zeta}}$, and $\bL^\bA_2$ leaves $E_\zeta$ invariant iff $\widetilde{\zeta}$ is collinear to $\zeta$. This happens if and only if $\zeta$ is an eigenstate of the harmonic oscillator, that is $\zeta = \varphi_n$ for some $n \in \N_0$.
\end{remark}

The three--dimensional analogue of the previous Proposition reads as follows.

\begin{theorem}\label{th:decomposition}
    Let $\gamma$ be a bounded operator on $L^2(\RR^3)$ commuting with all magnetic translations $\fm_\bR$ and with $\bL^\bA_2\otimes \bbI$.  Then, there exists a family $(\gamma_n)_{n\in\N}$ of bounded  operators on $L^2(\R)$, with $\| \gamma_n \|\le \| \gamma \|$, and so that 
    \begin{equation} \label{eq decomp eta}
        \gamma = \sum_{n=0}^\ii  \bP_n\otimes \gamma_n.
    \end{equation}
If $\gamma$ is a locally trace class operator, then its density depends only on $x_3$
    \[
    \rho_{\gamma}(\bx) = \rho_{\gamma}(x_3)  = \frac{b}{2 \pi} \sum_{n=0}^\infty \rho_{\gamma_n}(x_3).
    \]
\end{theorem}

\begin{proof}
    Let us consider two fixed test functions $\phi,\psi\in L^2(\RR)$, and define the operator $\eta_{\phi,\psi}$ on $L^2(\RR^2)$  by 
    $$
    \forall f,g\in L^2(\RR^2), \quad \langle f, \eta_{\phi,\psi}, g\rangle_{L^2(\RR^2)}:= \langle f\otimes \phi, \gamma (g\otimes \psi) \rangle_{L^2(\R^3)}.
    $$
    The conditions on $\gamma$ imply that $\eta_{\phi,\psi}$ is a bounded self--adjoint operator that commutes with all $\fm_\bR$ and with $\bL_\bA$. Thus, using Proposition~\ref{prop:decomposition}, $\eta_{\phi,\psi}$ can be decomposed as 
    $$
    \eta_{\phi,\psi}=\sum_{n\in\N_0} \lambda_n(\phi, \psi) \bP_n.
    $$
    Since $\eta$ is a bounded operator, we have for any normalized  $\Phi_n \in L^2(\R^2)$ in the range of $\bP_n$,
    \[
    \left| \lambda_n(\phi, \psi) \right| = \left| \langle \Phi_n \otimes \phi, \gamma (\Phi_n \otimes \psi) \rangle_{L^2(\R^3)} \right| \le \| \gamma \|_{\rm op}  \| \psi \|_{L^2(\R)}\norm{\phi}_{L^2(\R)}.
    \]
    We deduce that the map $(\phi, \psi) \mapsto \lambda_n(\phi, \psi)$ is sesquilinear and bounded, with bound smaller than $ \| \eta \|_{\rm op}$. The result then follows by taking $\gamma_n$ the bounded self-adjoint operator on $L^2(\RR)$ defined by
    $$
    \langle \phi, \gamma_n\psi \rangle:=  \lambda_n(\phi, \psi).
    $$
    Finally, for $\eta$ of the form~\eqref{eq decomp eta}, we obtain, using $\rho_{\bP_n} = \frac{b}{2 \pi}$,
    \[
    \rho_\eta(\bx) = \eta(\bx, \bx) = \sum_{n \in \N} \rho_{\bP_n}(x_1, x_2) \rho_{\gamma_n}(x_3) = \frac{b}{2 \pi} \sum_{n \in \N} \rho_{\gamma_n}(x_3).
    \]
\end{proof}

For an operator $\eta$ of the form~\eqref{eq decomp eta}, the {\em trace per unit-surface}, defined as the limit
\[
    \VTr( \eta ) = \lim_{L \to \infty} \frac{1}{L^2} \Tr \left( \1_{\Gamma_L} \eta \1_{\Gamma_L}  \right), \quad \Gamma_L = [-\tfrac{L}{2} , \tfrac{L}{2}]^2\times \R
\]
takes the simpler form
\[
    \VTr(\eta) = \frac{b}{2 \pi} \sum_{n=0}^\infty \Tr_1(\eta_n),
\]
where we have used that the density of any Landau level is $\rho_{\bP_n} = \frac{b}{2 \pi}$.


\section{Reduction of the kinetic energy, and applications}
\label{sec:reduction}

We now exploit the particular structure of states $\gamma \in \cP^\bA$ to deduce their kinetic energy.

\subsection{Reduction of the kinetic energy}
Recall that the set $\cP^\bA$ has been defined in~\eqref{eq:def:PA} as the set of one-body density matrices $\gamma$, satisfying the Pauli principle $0 \le \gamma \le 1$, which commute with all magnetic translations, and with the Landau operator $\bL^\bA_2 \otimes \bbI$. We also recall that the set of reduced states $\cG$ is defined by
\[
\cG :=  \left\{ G \in \cS(L^2(\R)), \quad G \ge 0\right\}.
\] 
The main result of this section is the following.
\begin{theorem}\label{th:reduction}
    For any $\gamma\in \cP^\bA$, there is an operator $G \in \cG$ satisfying $\rho_G=\rho_\gamma$ and 
    \begin{equation}\label{eq:inequality-kinetic}
        \frac12 \VTr(\bL^\bA_3 \gamma)  \geq \frac{1}{2}\Tr( - \Delta_1 G)+ \Tr(F(b,G)),
    \end{equation}
    where 
    \begin{equation}\label{eq f(b,g)}
        F(b,g) :=  \pi g^2 + \frac{b^2}{4 \pi } \left\{ \frac{2 \pi g}{b}   \right\} \left(1 -  \left\{ \frac{2 \pi g}{b}   \right\}  \right).
    \end{equation}
    Conversely, for any $G\in \cG$, there is  $\gamma\in \cP^\bA$ so that $\rho_\gamma=\rho_G$, and for which there is equality in~\eqref{eq:inequality-kinetic}.
    In particular, for any (representable) density $\rho$,
    \begin{equation}\label{Eq equat intro 2}
        \inf_{\gamma \in \cP^\bA \atop \rho_\gamma = \rho} \left\{ 	\frac12 \VTr(\bL^\bA_3 \gamma) \right\} 
        =  \inf_{G \in \cG \atop \rho_G = \rho} \left\{ \frac12 \Tr( - \Delta_1 G) + \Tr(F(b,G)) \right\}. 
    \end{equation}
\end{theorem}

\begin{proof}
    According to Theorem~\eqref{th:decomposition}, any $\gamma \in \cP^\bA$ can be decomposed as 
    \begin{equation}\label{eq decomposition gamma}
        \gamma = \sum_{n=0}^\infty \bP_n \otimes \gamma_n \quad \text{with} \quad \gamma_n \in \cS(L^2(\R)), \quad 0 \le \gamma_n \le 1.
    \end{equation}
    For a state of the form~\eqref{eq decomposition gamma}, we define the operator $G_\gamma \in \cS(L^2(\R))$ by
    \[
    G_\gamma := \frac{b}{2 \pi} \sum_{n=0}^\infty \gamma_n.
    \]
    Since $\gamma_n \ge 0$, we have $G_\gamma \ge 0$ as well, so $G_\gamma \in \cG$. Also, since $\rho_{\bP_n}(\bx) = \frac{b}{2 \pi}$, we deduce that $\rho_G = \rho_\gamma$.
    
    Recalling that $\bL^\bA_3 := \bL^\bA_2 \otimes \bbI_{L^2(\R)} + \bbI_{L^2(\R^2)} \otimes ( - \Delta_1)$, and using Proposition~\ref{prop spectral properties L_A}, we obtain that
    \begin{align*}
        \frac12 \VTr \left( \bL^\bA_3 \gamma \right)  & = \frac12 \sum_{n=0}^\infty \VTr \left( \varepsilon_n \bP_n \otimes \gamma_n \right) +  \frac12 \sum_{n=0}^\infty \VTr \left( \bP_n \otimes ( - \Delta_1 \gamma_n ) \right) \\
        & = \frac{b}{4 \pi} \sum_{n=0}^\infty \varepsilon_n \Tr(\gamma_n) + \frac{b}{4 \pi}\sum_{n=0}^\infty \Tr \left( - \Delta_1 \gamma_n  \right)\\
        & = \frac{b}{4 \pi} \sum_{n=0}^\infty \varepsilon_n \Tr(\gamma_n) + \frac12 \Tr \left( - \Delta_1 G  \right).
    \end{align*}
    The first term cannot be expressed directly as a function of $G$, but we have an inequality for this term. Since $G$ is a positive operator with finite trace, it is compact, and admits a spectral decomposition of the form $G=\sum_j g_j |\psi_j\rangle\langle\psi_j|$ with $g_j>0$ and $\sum_j g_j <\ii$. Evaluating the trace of $\gamma_n$ in the $\{ \psi_j\}$ basis, and changing the order of the sums (all terms are positive), we obtain
    \[
      \sum_{n=0}^\infty \varepsilon_n \Tr(\gamma_n) =\sum_{j=1}^\infty \left( \sum_{n=0}^\infty \varepsilon_n \left\langle \psi_j, \gamma_n \psi_j \right\rangle \right).
    \]
    Since $0 \le \gamma_n \le 1$, the quantity $m_j(n) := \left\langle \psi_j, \gamma_n \psi_j \right\rangle$ satisfies $0 \le m_j(n) \le 1$. In addition, we have $\frac{b}{2 \pi} \sum_{n} m_j(n) = \frac{b}{2 \pi}\langle \psi_j, \sum_n \gamma_n  \psi_j \rangle = \langle \psi_j, G \psi_j \rangle = g_j$. So we have the inequality
    \begin{equation}\label{eq:penalty}
        \sum_{n=0}^\infty \varepsilon_n \Tr(\gamma_n) \ge  \sum_{j=1}^\infty \inf_m \left\{ \sum_{n = 0}^\infty \varepsilon_n m(n), \ 0 \le m(n) \le 1, \ \sum_{n = 0}^\infty m(n) = \frac{2 \pi g_j}{b} \right\}.
    \end{equation}
Since the $\varepsilon_n$ are ranked in increasing order, we can apply the bathtub principle~\cite[Theorem 1.4]{lieb2001analysis}. The optimal $m$ for the above minimization is given by 
    \begin{equation}\label{eq:def-mj*}
     m_j^*(n) = \begin{cases}
        1 & \text{for all}\quad 0 \le n\le \lfloor\frac{2\pi g_j}{b}\rfloor - 1\\
        \{\frac{2\pi g_j}{b}\} & \text{for}\quad n=\lfloor\frac{2\pi g_j}{b}\rfloor \\
        0 & \text{otherwise}.
    \end{cases}. 
\end{equation}
We now calculate the infimum in the RHS  of~\eqref{eq:penalty} using the explicit formula of the optimal function $m^*$. Recalling that $\varepsilon_n = b(2n + 1)$ and denoting by $x := \frac{2 \pi g_j}{b}$,  we obtain 
    \begin{align*}
        \sum_{n = 0}^\infty\eps_n m_j^*(n) & = b \sum_{n=0}^{\lfloor x\rfloor - 1} (2n + 1) + b \left(2 \left\lfloor x \right\rfloor+1\right) \left\{  x \right\} \\
         & = b \left( x^2 + \{ x \}(1 - \{ x \})  \right) = \frac{4 \pi }{b} F(b, g_j),
    \end{align*}
    with the function $F$ defined in~\eqref{eq f(b,g)}. Summing in $j$ and gathering the terms gives the inequality
    \[
        \frac12 \VTr \left( \bL^\bA_3 \gamma \right)   \ge \Tr \left(  F(b, G) \right) + \frac12 \Tr \left( - \Delta_1 G  \right),
    \]
    which proves the first part of the Theorem. Conversely, given $G = \sum_{j} g_j | \psi_j \rangle \langle \psi_j | \in \cG$, we  consider the state
    \[
        \gamma^* := \sum_{n=0}^\infty \bP_n \otimes \gamma_n^*, \quad \text{with} \quad \gamma_n^* := \sum_{j=1}^\infty m_j^*(n) | \psi_j \rangle \langle \psi_j |
    \]
    and $m_j^*$ defined as in~\eqref{eq:def-mj*}. The operator $\gamma^*$ belongs to $\cP^\bA$, satisfies $G_{\gamma^*} = G$, and gives an equality in~\eqref{eq:inequality-kinetic}.
\end{proof}

\subsection{Some properties of the function F}

Let us collect some useful properties of the function $F$. A plot of $F$ is displayed in Figure~\ref{fig:plotF} below.

\begin{proposition}\label{prop:props-F}
    The function $F$ in~\eqref{eq f(b,g)} is continuous and satisfies
    \begin{equation}\label{eq:encadrement-F}
        \pi g^2 \le F(b, g) \le \pi g^2 + \frac{b^2}{16 \pi},
    \end{equation}
with equality in the left for $\frac{2\pi g}{b}\in\N$, and equality in the right for $\frac{2\pi g}{b}\in\N+\frac12$, and $F(b,g)\to \pi g^2$ as $b\to 0$.    For any $b\geq 0$,  the map $g \mapsto F(b, g) - \pi g^2$ is $\frac{b}{2 \pi}$ periodic and the map $g \in \RR_+\mapsto F(b, g)$ is piece-wise linear, increasing and convex. Finally, for all $0\leq g < \frac{b}{2 \pi}$, we have $F(b, g) = \frac12 b g$ 
\end{proposition}
\begin{proof}
    The first part is straightforward from the definition~\eqref{eq f(b,g)}. To see that it is convex, piece-wise linear and increasing, we use the alternative form 
    \begin{equation} \label{eq: F(b,g) alternate form}
        F(b, g) = \frac{b^2}{4 \pi} \left( x(1 + 2 \lfloor x \rfloor) -  \lfloor x \rfloor - \lfloor x \rfloor^2 \right),
    \end{equation}
where we have denoted by $x := \frac{2 \pi g}{b}$. 
    When  $0\leq g < \frac{b}{2 \pi}$, which corresponds to  $0 \leq  x < 1$, $F(b, g) = \frac12 b g$. 
\end{proof}

\begin{remark}\label{rmk F(b,G) vs pi Tr(G^2)}
	The left inequality of~\eqref{eq:encadrement-F} implies $\Tr(F(b,G)) \ge \pi \Tr(G^2)$, hence, together with~\eqref{eq:equat_intro_1}, that
    \[
        \inf_{\gamma \in \cP^\bA \atop \rho_\gamma = \rho}  \left\{ \frac12 \VTr \left( \bL^\bA_3 \gamma   \right) \right\} \ge \inf_{\gamma \in \cP \atop \rho_\gamma = \rho} \left\{ \frac12 \VTr \left( -\Delta_3 \gamma   \right)\right\}.
    \] 
    In particular, the kinetic energy is higher with the magnetic field. This is a kind of diamagnetic inequality for 2D materials.
\end{remark}

\begin{remark}
     The fact that $F(b,g)\to\pi g^2$ as $b\to0$, for all $g\in\R^+$, means that, our reduction approach in this manuscript is coherent with the one without magnetic field already treated in \cite{GLM21}.
\end{remark}

\begin{remark}
     Splitting $F$ into $F(b,G)=\pi g^2 + \tilde{F}(b,g)$, we see that the effect of adding a magnetic field $\bB=(0,0,b)$ is a periodic perturbation of the energy with no magnetic field.
\end{remark}


\subsection{Reduced DFT models}
\label{ssec:rHF}

In the context of DFT, the previous result suggests modelling the electronic state in a homogeneous slab of charge distribution $\mu (\bx) = \mu(x_3)$ under a constant magnetic field $\bB = b(0, 0, x_3)$ by a reduced state $G\in \cG$ whose energy per unit surface is given by 
\begin{equation}\label{eq DFT model}
\cE(G)= \frac12\Tr(-\Delta_1 G)+\Tr\left( F_b( G) \right) + \frac12 \cD_1(\rho_G -\mu) + E^{\rm xc}(\rho_G). 
\end{equation}
Here, $E^{\rm xc}$ models is an exchange-correlation energy {\em per unit surface}, and $\cD_1$ is the one--dimensional Hartree term. This last term has been extensively studied in our previous work~\cite[Section 3.1]{GLM21}, and is defined as follows. For $f \in \cC := \left\{  f \in L^1(\R), \ W_f \in L^2(\R) \right\}$, where $W_f(x) := \int_{-\infty}^x f$ is a primitive of $f$, we have
\[
    \cD_1(f) := 4 \pi \int_{\R} | W_f |^2(x) \rd x.
\]
We have proved in~\cite[Proposition 3.3]{GLM21} that the elements $f \in \cC$ have null integral $\int f = 0$, that the map $\cC \ni f \mapsto \cD_1(f)$ is convex, and that
\[
    \cD_1(f) = 4 \pi \iint_{(\R_+)^2 \times (\R_-)^2} \min \{ | x |, | y | \} f(x) f(y) \rd x \rd y = \int_{\R} \Phi_f(x) f(x) \rd x,
\]
with the mean-field potential
\[
    \Phi_f(x) := 4 \pi \int_{\R^\pm} \min \{ | x |, |y | \} f(y)  \rd y, \quad x \in \R^\pm.
\]
The function $\Phi_f$ is continuous, and is the (unique) solution to 
\[
    - \Phi_f''(x) = 4 \pi f , \quad \Phi_f'(x) \xrightarrow[x \to \pm \infty]{} 0, \quad \Phi_f(0) = 0.
\]
In practice, we restrict the minimization problem to the $G$ so that $\rho_G - \mu \in \cC$. This implies in particular the neutrality condition $\Tr(G) = \int \rho_G = \int \mu = \nu$. On the other hand, if $G$ is a trace class operator, then $\rho_G \in L^1$ (we assume that $\mu \in L^1$ as well), and if $\rho_G - \mu$ has null integral, then $\cD_1(\rho_G - \mu) < \infty$ iff $\rho_G - \mu \in \cC$. 

Note that there is no Pauli principle on $G$ for admissible states $\cG$. It has been replaced by a penalization term $+ F(b, G)$ in the energy.

\begin{remark}
    The energy~\eqref{eq DFT model} is obtained when minimizing a 3-dimensional DFT model over transitionally invariant states. In particular, this model does not include possible spatial symmetry breaking along the first $2$ variables. Such phenomena are known to exist in two-dimensional electron gas under magnetic field due to the de Haas--van Alphen effect~\cite{de1930dependence}. In some real-life systems {\em e.g.} Br$_2$, magnetic domains form, sometimes called Condon domains~\cite{condon1966nonlinear, markiewicz1985giant}. Our simple model is unable to capture these effects.
\end{remark}

\subsection{The reduced Hartree--Fock case}
Let us illustrate the previous discussion in the particular case of the reduced Hartree-Fock (rHF) model, in which $E^{\rm xc} = 0$ in~\eqref{eq DFT model}. We let $0\le\mu \in L^1(\R)$ be a nuclear density describing a homogeneous 2D material and denote by $\nu=\int_\R \mu$ the total charge per unit surface.

We denote by 
\begin{equation}
    \cE_b^\rHF(G) := \frac12\Tr(-\Delta_1 G) +\Tr(F(b,G)) +\frac12 \cD_1(\rho_G -\mu) 
\end{equation}
the corresponding rHF energy per unit surface, and study the minimization problem
\begin{equation*}
    \cI_b^\rHF :=\inf\set{\cE_b^\rHF(G),\;  G\in\cG^\nu}, \quad \text{with} \quad
    \cG^\nu := \left\{ G \in \cG, \ \Tr(G) = \nu \right\}.
\end{equation*} 

Following the exact same lines as~\cite[Theorem 2.7]{GLM21}, one has the following. 
\begin{theorem}
    The problem $\cI_b^\rHF$ admits a minimizer, and all minimizers share the same density.
\end{theorem}
We skip the proof for brevity. The uniqueness of the density comes from the fact that the problem is strictly convex in $\rho_G$. However, unlike the case without magnetic field,  $F(b,\cdot)$ is not strictly convex for $b > 0$. It is unclear to us whether the minimizer of $\cI_b^\rHF$ is unique.

\medskip

We would like to write the Euler-Lagrange equations for a minimizer $G_*$. Recall that $g \mapsto F(b, g)$ is continuous and convex (but is not smooth). We denote by $f_b := \partial_g F(b,\cdot)$ its subdifferential, a set-valued function defined by
\[
    f_b(g) := \left\{ a \in \R, \quad \forall g' \in \R, \quad F(b,g') - F(b,g) \ge a (g' - g) \right\}.
\]
In our case, since the function $F_b$ is piece-wise linear, $f_b$ is explicit. From~\eqref{eq: F(b,g) alternate form}, we obtain (in the following lines, $\{ a \}$ denotes the singleton $a$)
\[
    f_b(g) = \begin{cases}
        \left\{\frac{b^2}{4 \pi} (2n + 1)\right\}  \quad \text{if } \exists n \in \N_0,\;  \quad n < \frac{2 \pi g}{b} < n+1 \\
        \left[ \frac{b^2}{4 \pi} (2n - 1), \frac{b^2}{4 \pi} (2n + 1)  \right] \quad \text{if} \quad \frac{2 \pi g}{b} = n \in \N_0.
    \end{cases}
\]

Its inverse map, noted $h_b$, is the set-valued function so that $y \in f_b(x)$ iff $x \in h_b(y)$. One finds, for $y > 0$,
\[
    h_b(y) =  \begin{cases}
        \left[ n \frac{b}{2 \pi}, (n+1)  \frac{b}{2 \pi} \right] \quad \text{if} \quad n := \frac12 \left( \frac{4 \pi}{b^2} y - 1  \right) \in \N_0 , \\
        \left\{ n \frac{b}{2 \pi } \right\} \quad \text{if there is $n \in \N_0$ so that} \quad n -1<  \frac12 \left( \frac{4 \pi}{b^2} y - 1 \right)< n.
    \end{cases}
\]
We extend the definition of  $h_b$  by setting $h_b(y) = 0$ for $y < 0$.

In order to work with functions, it is useful to introduce the maps 
\begin{equation} \label{eq:def:fpm}
    f_b^\pm(g) := \lim_{t \to 0^\pm} \frac{1}{t}\left( F(b, g + t) - F(b, g) \right)
\end{equation}
so that $f_b(g) = [f_b^-(g), f_b^+(g)]$ for all $g \in \R^+$. Of course, if $g \notin \frac{b}{2 \pi} \N_0$ is a regular point, then $f_b(y) = f_b^+(y) = f_b^-(y)$. We define the maps $h_b^\pm$ similarly, so that $h_b(y) = [h_b^-(y), h_b^+(y)]$ for all $y \in \R$.

\medskip

The Euler--Lagrange equations for $G_*$ takes the following form (see end of the section for the proof).
    \begin{proposition}[Euler-Lagrange equations] \label{lem:ELeqt}
        Let $G_*$ be a minimizer of $\cI^{\rHF}_b$. Then there is $\lambda \in \R$ so that
        \begin{equation}\label{eq Euler Lagrange}
            \begin{cases}
                h_b^-(\lambda - H_*)  \le G_* \le h_b^+ (\lambda - H_*) \\
                H_*  := - \frac12 \Delta_1 + \Phi_* \\
                -\Phi_*''  = 4 \pi( \rho_* - \mu ), \quad \Phi_*'(x) \xrightarrow[x \to \pm \infty]{} 0, \quad \Phi_*(0) = 0,
            \end{cases}
        \end{equation}
    where $\rho_* = \rho_{G_*}$ is the associated density of $G_*$,  and $\Phi_*$ is the mean-field potential, defined as the unique solution of the last equation.
    \end{proposition}

The first equation can also be written as
 \[
    G_* \in h_b(\lambda - H_*),
 \]
 and means that if $G_* = \sum_j g_j | \psi_j \rangle \langle \psi_j |$ is the spectral decomposition of the optimizer, then $\psi_j$ is also an eigenfunction of $H_*$ for an eigenvalue $\varepsilon_j$ so that
\[
    g_j \in h_b(\lambda - \varepsilon_j), \quad \text{or equivalently} \quad \varepsilon_j \in \lambda - f_g(g_j).
\]
Conversely, if $\varepsilon < \lambda$ is an eigenvalue of $H_*$, then $\varepsilon = \varepsilon_j$ for some $j$.

\medskip

In practice, for numerical purpose, one rather considers an approximation $F^\delta_b$ of $F$, which is smooth, strictly convex, and so that $\| F^\delta_n - F \|_\infty < \delta$. In this case, one can repeat the arguments in~\cite[Theorem 2.7]{GLM21}, and the first line of~\eqref{eq Euler Lagrange} becomes
\[
    \left(F_b^\delta\right)'(G_*) = \lambda- H_*, \quad \text{or, equivalently} \quad G_* = \left[ \left(F_b^\delta\right)'  \right]^{-1} (\lambda- H_*).
\]

\begin{remark}[Strong magnetic fields]
    In the case where $b > 2 \pi \nu$, any $G \in \cG^\nu$are positive and satisfies $\Tr(G) = \nu$, hence all eigenvalues of $G$ are smaller than $\nu$. In particular, 
    \[
        F_b(G) = \frac12 b G, \quad \text{hence} \quad \Tr \left( F_b (G) \right) = \frac12 b \nu
    \]
    is a constant, independent of $G \in \cG^\nu$. In this case, $G_*$ is also the minimizer of
    \[
        \inf \left\{ \frac12 \Tr( - \Delta_1 G) + \frac12 \cD_1(\rho_G - \mu ), \quad G \in \cG^\nu \right\}.
    \]
    This minimizer is therefore independent of $b > 2 \pi \nu$, reflecting the fact that all electrons lie in the lowest Landau level. Following the previous lines, we deduce that $G_*$ is a rank-$1$ operator, of the form $G_* = \nu | \psi_* \rangle \langle \psi_*|$, with $\psi_*$ minimizing
    \[
        \inf \left\{ \frac \nu2 \int_{\R} | \nabla \psi_* |^2 + \frac12 \cD_1\left( \nu | \psi_* |^2 - \mu \right), \quad \psi_* \in L^2(\R), \ \| \psi_* \| = 1 \right\}.
    \]
\end{remark}

\begin{proof}[Proof of Proposition~\ref{lem:ELeqt}]
  Let $G_*$ be a minimizer of $\cI^\rHF_b$, and let $\rho_*$ and $\Phi_*$ be the corresponding density and mean-field potential, and set $H_* := - \frac12 \Delta_1 + \Phi_*$. Recall that $\rho_*$ (hence $\Phi_*$ and $H_*$) is uniquely defined. 
    
    First, we claim that $G_*$ commutes with $H_*$. This is a standard result in the case where the map $F$ is smooth (say of class $C^1$), using that 
    \[
        \Tr ( F (G_* + H)) = \Tr(F(G_*)) + \Tr(F'(G_*) H ) + o(H).
    \]
    In our case however, the map $F$ is only piece-wise smooth, and we need a direct proof. Let $A$ be a finite-rank symmetric operator on $L^2(\R)$, and set $G_{t, A} := \re^{ - \ri t A} G_* \re^{ \ri t A}$. Since $G_{t, A}$ is a unitary transformation of $G_*$, we have $\Tr (G_{t, A}) = \Tr(G_*) = \nu$, that is $G \in \cG^\nu$ and $\Tr( F(b, G_{t, A})) = \Tr ( F(b, G_*))$. In particular, 
    \begin{align*}
        \cE^\rHF_b(G_{t, A}) & = \cE^\rHF_b(G_*) \\
        & + \frac12 \left( \Tr(- \Delta_1 [G_{t, A} - G_*]) +  \cD_1 \left( \rho_{G_{t, A}} - \mu\right) -  \cD_1 \left( \rho_{*} - \mu\right) \right).
    \end{align*}
    Together with the fact that
    \[
        G_{t, A} = G_* + \ri t \left[ G_*, A \right] + o(t), 
    \]
    and the definition of $H_*$, we deduce that
    \[
        \cE^\rHF_b(G_{t, A})  = \cE^\rHF_b(G_*)  + \ri t \Tr \left( H_*   [G_*, A]   \right) + o(t).
    \]
    Since the minimum of $\cE^\rHF_b$ is obtained for $t = 0$, the linear term in $t$ must vanish, that is:
    \begin{align*}
        0 & = \Tr \left( H_*  [G_*, A] \right) = \Tr \left( H_* G_* A - H_* A G_* \right) \\
         & = \Tr \left( H_* G_* A - G_* H_* A  \right) = \Tr \left( [H_*, G_*] A \right),
    \end{align*}
    where we have used cyclicity of the trace, and the fact that $A$ is finite rank (so all operators are trace-class). Since this is true for all finite rank symmetric operators $A$, we deduce as wanted that $[H_*, G_*] = 0$.
    
    \medskip
    
    Recall that $G_*$ is positive compact (even trace class). So there is $M := {\rm rank}(G_*) \in \N \cup \{ \infty \}$ and an orthonormal family $\{ \psi_j \}_{j \in [1, M]}$ so that
    \[
        G_* = \sum_{j = 1}^M g_j | \psi_j \rangle \langle \psi_j |, \quad \text{with} \quad H \psi_j = \varepsilon_j \psi_j,
    \]
    and with $g_1 \ge g_2 \ge \cdots \ge g_M > 0$. The orthonormal family $\{ \psi_j\}_{1 \le j \le M}$ spans $\Ran(G_*)$. For two indices $(i,j) \in [1, M]$, we consider the operator
    \[
        G_t^{(i,j)} := G_* + t \left( | \psi_i \rangle \langle \psi_i | - | \psi_j \rangle \langle \psi_j | \right).
    \]
    For $t$ small enough ($| t | < \min \{ g_i, g_j \}$), the operator $G_t^{(i,j)}$ is positive, and with $\Tr(G_t^{(i,j)}) = \Tr(G_*) = \nu$, hence $G_t^{(i,j)} \in \cG^\nu$. In addition, we have
    \begin{align*}
        \cE^\rHF_b(G_t^{(i,j)})  -  \cE^\rHF_b(G_*) =  &  \left( F(b, g_i + t)  - F(b, g_i) \right) +  \left( F(b, g_j - t)  - F(b, g_j) \right) \\
        & + t \underbrace{\Tr (H_* ( | \psi_i \rangle \langle \psi_i | - | \psi_j \rangle \langle \psi_j | ))}_{ = \varepsilon_i - \varepsilon_j} + o(t).
    \end{align*}
    By optimality of $G_*$, this quantity is always positive.
    Taking the limit $t \to 0^+$, and recalling the definition of $f_b^\pm$ in~\eqref{eq:def:fpm}, we deduce that
    \[
        f_b^+(g_i) +  \varepsilon_i \ge f_b^-(g_j) +  \varepsilon_j.
    \]
    This inequality is valid for all $1 \le i,j \le M$, so
    \[
        \inf_{1 \le i \le M}  \left(  f_b^+(g_i) +  \varepsilon_i  \right) \ge \sup_{1 \le j \le M} \left( f_b^-(g_j) +  \varepsilon_j \right) =: \lambda
    \]
     For all $1 \le j \le M$, we have
    \[
        f_b^-(g_j) +  \varepsilon_j \le \lambda \le f_b^+(g_j) +  \varepsilon_j, \quad \text{hence} \quad \lambda - \varepsilon_j \in [f_b^-(g_j), f_b^+(g_j)].
    \]
    This is also $\lambda - \varepsilon_j \in f_b(g_j)$, or equivalently $g_j \in h_b(\lambda - \varepsilon_j)$. Note that if there is an eigenvalue $g_j \notin \frac{b}{2 \pi} \Z$, then $f_b^-(g_j) = f_b^+(g_j)$ for this eigenvalue, and there is equality. In other words, we have proved that 
    \[
        G_* \in h_b (\lambda - H_*) \quad \text{on} \quad \Ran(G_*).
    \]
    It remains to prove the result on $\Ker(G_*)$. Let $\psi \in \Ker(G_*)$. This time, we consider the perturbed state
    \[
        G_t^{(j)} := G_* + t \left( | \psi \rangle \langle \psi | - | \psi_j \rangle \langle \psi_j | \right),
    \]
    which is in $\cG^\nu$ for all $0 \le t \le g_j$. Taking the limit $t \to 0^+$, and reasoning as before, we get
    \[
        f_b^+(0) + \langle \psi, H_* \psi \rangle \ge f_b^-(g_j) + \varepsilon_j, 
        \quad \text{hence} \quad
         f_b^+(0) + \langle \psi, H_* \psi \rangle \ge \lambda,
    \]
    where we took the supremum in $1 \le j \le M$ in the last inequality. We deduce that, for all $\psi \in \Ker(G_*)$, 
    \[
        f_b^+( \langle \psi, G_* \psi \rangle )  = f_b^+(0) \ge \lambda - \langle \psi, H_* \psi \rangle, 
    \]
    so 
    \[
        f_b^+(G_*)  \ge \lambda - H_*, \quad \text{on} \quad \Ker(G_*).
    \]
    Together with the fact that $f_b^-(G_*) = 0$ on $\Ker(G_*)$, we obtain $G_* \in h_b(\lambda - H_*)$ on $\Ker(G_*)$ as well.
\end{proof}




\section{Models with spin}
\label{sec:spins}

In this section, we explain how to extend our results to the case where the spin is taken into account. In this case, the density matrix is an operator $\gamma \in\cS(L^2(\R^3,\C^2))$ satisfying the Pauli-principle $0 \le \gamma \le 1$. Such an operator can be  decomposed as a  $2\times 2$ matrix of the form
\[
{\gamma}=\begin{pmatrix}
    \gamma^{\uparrow\uparrow} & \gamma^{\uparrow\downarrow} \\
    \gamma^{\downarrow\uparrow} & \gamma^{\downarrow\downarrow}
\end{pmatrix}.
\]
The spin--density $2 \times 2$ matrix of $\gamma$ is $R_\gamma(\bx) = \gamma(\bx, \bx)$ (see~\cite{Gontier2013} for details), and the total density is $\rho_\gamma(\bx) = R_\gamma^{\uparrow \uparrow}(\bx) + R_\gamma^{\downarrow \downarrow}(\bx) =  \gamma^{\uparrow \uparrow}(\bx, \bx) + \gamma^{\downarrow \downarrow}(\bx, \bx)$.

The kinetic energy operator is now the Pauli operator
$$
\bP^\bA_3 := \left[ \sigma\cdot \bra{-\ri \nabla + \bA } \right]^2,
$$
where $\sigma$ contains the Pauli matrices. In the case of constant magnetic field $\bB = (0, 0, b)$ with the gauge $\bA = b (0, x_1, 0)$, this operator becomes
\[
\bP^\bA_3 = \bL^\bA_3 \bbI_2 + \begin{pmatrix}
    b & 0 \\ 0 & -b
\end{pmatrix} =  \bL^\bA_2 \bbI_2 + \begin{pmatrix}
b & 0 \\ 0 & -b
\end{pmatrix} - \partial_{x_3 x_3}^2 \bbI_2.
\]
In what follows, we denote by $\cB := \begin{pmatrix}
    b & 0 \\ 0 & -b
\end{pmatrix}$. This term corresponds to the Zeeman term. There are several ways to read this operator. Indeed, we have
\begin{align}
	L^2(\R^3, \C^2) & \simeq L^2(\R^3) \otimes \C^2 \simeq L^2(\R^2) \otimes L^2(\R) \otimes \C^2 \nonumber \\
	& \simeq L^2(\R^2)  \otimes L^2(\R, \C^2) \label{eq:splitting1}\\
	& \simeq L^2(\R^2, \C^2)  \otimes L^2(\R). \label{eq:splitting2}
\end{align}

In the decomposition~\eqref{eq:splitting1} (resp.~\eqref{eq:splitting2}), we split the Pauli operator as 
\[
    \bP^\bA_3  = \bL^\bA_2 \otimes \bbI + \bbI \otimes (\cB - \partial_{x_3 x_3}^2 ), \quad
    (\text{resp. } \quad
    \bP^\bA_3  = (\bL^\bA_2 + \cB)\otimes \bbI + \bbI \otimes (- \partial_{x_3 x_3}^2 ).)
\]

\underline{First splitting: keeping the spin structure.}
This splitting is useful if one wants to keep track of the spin--density structure of $R_\gamma$ as a function of $x_3$. This happens for instance when using exchange--correlations functionals which depend explicitly on the spin, as in the LSDA model~\cite{von1972local, gontier2014existence}. In this case, we consider one-body density matrices $\gamma$ of the form
\begin{equation} \label{eq:gamma-spin1}
	\gamma = \sum_{n = 0}^\infty \bP_n \otimes \gamma_n, \quad \gamma_n \in \cS(L^2(\R, \C^2)), \ 0 \le \gamma_n \le 1.
\end{equation}
This is similar to what we have studied in the present article, but the $\gamma_n$ operators now act on $L^2(\R, \C^2)$. Introducing the operator $G := \frac{b}{2 \pi} \sum \gamma_n$ and the set
\[
\cG^{\rm spin} := \left\{  G \in L^2(\R, \C^2), \ G \ge 0 \right\},
\]
we obtain as before that, for all representable spin--density $2\times 2$ matrix $R$, we have
\begin{align*}
	& \inf_{\gamma  \ \text{of the form~\eqref{eq:gamma-spin1}}  \atop R_\gamma = R} \left\{ 	\frac12 \VTr(\bP^\bA_3 \gamma) \right\} 
	 \\
	& \qquad = \inf_{G \in \cG^{\rm spin} \atop R_G = R} \left\{ \frac12 \Tr( - \Delta_1 G) + \Tr(F(b,G)) \right\} + b \int_{\R} (R^{\uparrow \uparrow} - R^{\downarrow \downarrow}) .  \nonumber
\end{align*}
The last term is the Zeeman term, where we have used that $\Tr(\cB G) = b \Tr(G^{\uparrow \uparrow} - G^{\downarrow \downarrow}) = b \int (R^{\uparrow \uparrow} - R^{\downarrow \downarrow})$. 

\medskip

\underline{Second splitting: loosing the spin structure.}
If one is only interested in keeping the total density $\rho_\gamma$ instead of the full spin--density $2 \times 2$ matrix $R_\gamma$, one should instead use the spectral decomposition of the operator $\bP^\bA_2 := \bL^\bA_2 + \cB$. This one is easily deduced from the one of $\bL^\bA_2$ in Proposition~\ref{prop spectral properties L_A}. Using that $\varepsilon_n = b(2n + 1)$, we obtain $\bP^\bA_2 = \sum_{n=0}^\infty \widetilde{\varepsilon}_n \widetilde{\bP}_n $ with
\[
\widetilde{\varepsilon}_n = 2n b, 
\quad \widetilde{\bP}_0 = \begin{pmatrix}
	0 & 0 \\ 0 &  \bP_0
\end{pmatrix}, 
\quad  \widetilde{\bP}_n = \begin{pmatrix}
	\bP_{n-1} & 0 \\ 0 & \bP_{n} . 
\end{pmatrix} \text{ for }n\geq 1 .
\]
The lowest Landau level has now energy $\widetilde{\varepsilon}_0 = 0$, and is only occupied by spin-down electrons. The corresponding eigenspace has density $\rho_{\widetilde{\bP_0}} = \frac{b}{2 \pi}$. The other Landau levels are <<doubly>> occupied, with density $\rho_{\widetilde{\bP_n}} = \frac{b}{\pi}$.

\medskip

In this case, we consider density matrices $\gamma$  of the form 
\begin{equation} \label{eq:gamma-spin2}
	\gamma = \sum_{n = 0}^\infty \widetilde{\bP_n} \otimes \widetilde{\gamma}_n, \quad \widetilde{\gamma}_n \in \cS(L^2(\R, \C)), \ 0 \le \widetilde{\gamma}_n \le 1.
\end{equation}
\begin{remark}
    One can prove that the set of such states is smaller than the set of states $\gamma$ commuting with $\bP^\bA_2$ and the magnetic translations $\fm_\bR$. This comes from the fact that the family $(\fm_\bR)$ is not irreducible on $ \widetilde{E}_n := {\rm Ran} \widetilde{P_n}$. 
\end{remark}

The conclusion of Theorem~\ref{th:reduction} still holds in the spin setting, and the proof is similar with only minor modifications: we now consider the $G$ matrix defined by
\[
G_\gamma = \frac{b}{2 \pi} \left( \widetilde{\gamma}_0 + 2\sum_{n=1}^\infty \widetilde{\gamma}_n \right) \qquad \in \cG,
\]
and the optimal $m_j^*$ functions are given by $m_j^*(0) = \{ \frac{2 \pi g_j}{b}\}$ and $m_j^*(n) = 0$ if $\frac{2 \pi g_j}{b} < 1$, and, if $\frac{2 \pi g_j}{b} \ge 1$,
\[
m_j^*(n) = \begin{cases}
	1 & \quad \text{for} \quad n \le \lfloor \frac{\pi}{b}g_j + \frac{1}{2} \rfloor - 1\\
	\left\{ \frac{\pi}{b}g_j + \frac{1}{2} \right\} & \quad \text{for} \quad n =\lfloor \frac{\pi}{b}g_j + \frac{1}{2} \rfloor \\
	0 & \quad \text{otherwise}.
\end{cases}
\] 
Denoting by $y = \frac{\pi}{b}g_j + \frac12$, this gives the energy
\begin{align*}
	\frac{b}{2 \pi}  \sum_{n = 0}^\infty \widetilde{\varepsilon} _n m_j^*(n) & = \frac{b^2}{2 \pi} \sum_{n=1}^{ \lfloor y \rfloor - 1} 2 n +  \frac{b^2}{\pi} \lfloor y \rfloor \left\{   y \right\} = \frac{b^2}{2 \pi} \left( y^2 - y + \{ y \} (1 - \{ y \}) \right).
\end{align*}
We obtain that for any representable density $\rho$.
\[
\inf_{\gamma  \ \text{of the form~\eqref{eq:gamma-spin2}}  \atop \rho_\gamma = \rho} \left\{ 	\frac12 \VTr(\bP^\bA_3 \gamma) \right\} 
=  \inf_{G \in \cG \atop \rho_G = \rho} \left\{ \frac12 \Tr( - \Delta_1 G) + \Tr(F^{\rm spin}(b,G)) \right\},
\]
with the new functional 
\[
F^{\rm spin}(b, g) =  \frac{}{}\frac{\pi}{2} g^2 -  \frac{b^2}{8 \pi} + \frac{b^2}{2 \pi} \left\{  \frac{\pi g}{b} + \frac12  \right\} \left(1 -  \left\{ \frac{\pi g}{b} + \frac12\right\} \right).
\] 
A plot of $F^{\rm spin}(b = 1, g)$ is displayed in Figure~\ref{fig:plotF}. The constant $\frac{\pi}{2}$ in the $\frac{\pi}{2} g^2$ term is the  two-dimensional Thomas-Fermi constant, when the spin of the electron is included. Again, we see that magnetic effects gives a correction to the Thomas-Fermi approximation, but this time, including the Zeeman term, the magnetic energy is always lower than the no-magnetic one, and we have
\begin{equation} \label{eq:encadrement_Fspin}
	\frac{\pi}{2} g^2 - \frac{b^2}{8 \pi} \le F^{\rm spin}(b, g) \le \frac{\pi}{2} g^2.
\end{equation}

\begin{figure}[!ht]
    \begin{subfigure}{0.45\textwidth}
        \includegraphics[width=1\textwidth]{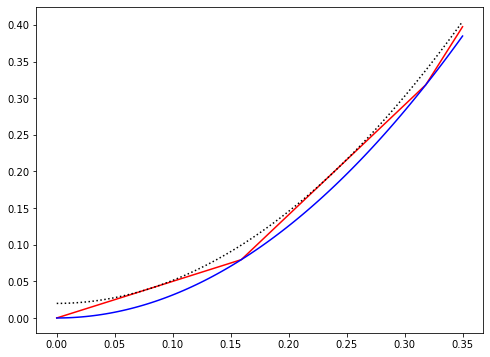}
    \end{subfigure}
    \begin{subfigure}{0.45\textwidth}
        \includegraphics[width=1\textwidth]{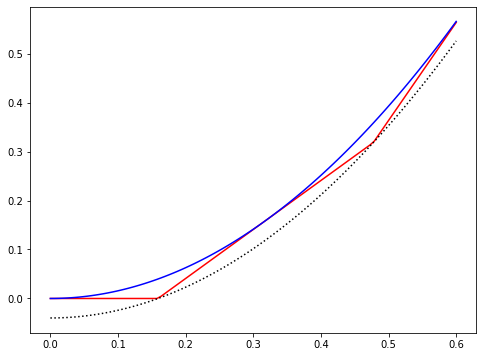}
    \end{subfigure}
    
    \caption{(Left). The map $g \mapsto F(b= 1, g)$ (red) together with its lower/upper bounds in~\eqref{eq:encadrement-F}. (Right) The map $g \mapsto F^{\rm spin}(b= 1, g)$ (red) and its bounds~\eqref{eq:encadrement_Fspin}.}
    \label{fig:plotF}
\end{figure} 

\bibliographystyle{plain}
\bibliography{biblio}

\end{document}